\documentclass[12pt]{article}
\usepackage[latin9]{inputenc}
\usepackage{geometry}
\usepackage{amsmath}
\usepackage{amssymb}
\usepackage{stmaryrd}
\usepackage{setspace}
\usepackage[authoryear]{natbib}
\usepackage[unicode=true,
 bookmarks=false,
 breaklinks=false,pdfborder={0 0 1},backref=section,colorlinks=false]
 {hyperref}
\hypersetup{
 colorlinks,citecolor=blue,pdftex}
\usepackage{breakurl}

\makeatletter

\usepackage{amsfonts}
\usepackage{amsthm}
\usepackage{lscape}
\usepackage{appendix}
\usepackage{dcolumn}
\usepackage{rotating}
\usepackage{comment}
\usepackage{color}

\usepackage{pgf}
\usepackage{tikz}\usepackage{caption}
\usepackage{subcaption}
\usepackage{tkz-tab}
\usepackage{tikz-3dplot}
\usetikzlibrary{calc}
\usetikzlibrary{arrows, automata}

\newcolumntype{d}[1]{D{.}{.}{#1}}
\newcolumntype{t}[1]{D{,}{,}{#1}}
\newcolumntype{i}[1]{D{.}{}{#1}}
\newtheorem{theorem}{Theorem}[section]

\theoremstyle{plain} 

\newtheorem*{asA}{Assumption A}

\makeatother

\begin{document}

\title{Comment: Individualized Treatment Rules\\
Under Endogeneity}

\author{Sukjin (Vincent) Han\thanks{\protect\href{mailto:vincent.han@bristol.ac.uk}{vincent.han@bristol.ac.uk}}\\
 Department of Economics\\
 University of Bristol}

\date{September 28, 2020}
\maketitle

\section{Heterogeneity, Individualized Treatments, and Endogeneity}

For the past few decades, scholars in statistical and social sciences
have acknowledged that heterogeneity is prevalent in the subjects
of their studies. The literature on individualized treatments is a
great example that embraces heterogeneity and develops fruitful policies
that benefit members of targeted populations. I congratulate Yifan
Cui, Eric Tchetgen Tchetgen, Hongxiang Qiu, Marco Carone, Ekaterina
Sadikova, Maria Petukhova, Ronald C. Kessler, and Alex Luedtke for
successfully pursuing this literature's endeavor. Their works---\citet{cui2019semiparametric},
\citet{qiu2020optimal}---especially stand out as they relax the
assumption of ``no unmeasured confounders'' that this entire literature
has relied on. In experimental and observational studies, there are
ample examples where endogeneity in treatment decisions cannot be
solely captured by observables, i.e., measured confounders. In a sense,
allowing for treatment endogeneity becomes relevant in a world of
heterogeneity, because typically, unobserved heterogeneity is a core
factor of individuals decisions to receive treatments. Therefore,
in my view, considering endogeneity in developing individualized treatments
is quite natural. 

I appreciate what \citet{cui2019semiparametric} and \citet{qiu2020optimal}
have achieved. Both use instrumental variable (IV) methods, but each
of them makes unique contributions. Under a set of identifying conditions,
\citet{cui2019semiparametric} derive a simple, closed-form expression
for the counterfactual mean based on which the (non-stochastic) optimal
regime is identified. The most notably result is Theorem 2.2, in which
the optimal regime is identified even without observing the endogenous
treatment. This is especially relevant in experimental settings where
partial compliance is suspected but the information about treatment
decisions is missing, e.g., due to the cost of fidelity assessments.
For example, the experimental dataset in \citet{murphy2001marginal}
has information from fidelity checks, but only aggregated information
(i.e., the compliance rate) is available. \citet{qiu2020optimal}
consider regimes that are allowed to be stochastic, which is pertinent
to the existence of budget constraints. They target not only optimal
treatments but also optimal encouragements, which is a sensible idea
given the use of the IV framework. With these notable features, they
proceed by employing the main identifying conditions similar to \citet{cui2019semiparametric}.
I appreciate that they develop the asymptotic distributions for the
estimated gains from optimal regimes, which enables inference.

I read the two papers with great pleasure. In this note, I discuss
the main identifying conditions commonly employed in the papers. Inspired
by their analyses, I propose alternative identifying conditions. Then,
motivated from these discussions, I turn the focus from the point
identification to partial identification. Along the way, I mention
two closely related and complementing papers on optimal dynamic treatment
regimes in longitudinal settings written by myself. I conclude by
making remarks on possible future directions for research.

\section{About Main Identifying Conditions}

In the presence of heterogeneity and endogeneity, identifying the
average treatment effects (ATEs), and thus optimal treatment regimes,
is a challenging task. To overcome this challenge, \citet{cui2019semiparametric}
and \citet{qiu2020optimal} utilize identifying conditions (Assumptions
7 and 8 in the former, Conditions A5b-1(b) and 2(b) in the latter)
that are introduced in \citet{wang2018bounded}. Since the identification,
estimation, and inference of the two papers rely on these conditions,
I would like to discuss them in depth here. In the next section, I
propose alternative identifying conditions.

For convenience, I follow the notation and labels of \citet{cui2019semiparametric};
in the parenthesis, I indicate the labels of \citet{qiu2020optimal}.
I attempt to interpret Assumptions 7 and 8 (Conditions A5b-1(b) and
2(b)) and understand their implications, starting from Assumption
8 (Condition A5b-2(b)). To facilitate the discussion, we maintain
Assumption 10 (Condition A6a) that $Z\perp(A_{z},Y_{a})|L$ (causal
IV). Then, Assumption 8, which the authors call ``independent compliance
type,'' can be expressed as
\begin{align*}
E[A_{1}-A_{-1}|L,U] & =E[A_{1}-A_{-1}|L].
\end{align*}
This expression reveals that the compliance type is defined by the
response to IV, i.e., $A_{1}-A_{-1}$, and that the assumption concerns
conditional mean independence between $A_{1}-A_{-1}$ and $U$. Since
$A_{1}-A_{-1}$ is not binary, $A_{1}-A_{-1}$ and $U$ can be associated
through higher order moments. Nonetheless, $A_{1}-A_{-1}$ takes only
three values ($A_{1}-A_{-1}\in\{-1,0,1\}$), and thus the mean independence
would still significantly restrict the joint distribution, yielding
a specific form of treatment endogeneity. At least to an econometrician
like myself, this is an unfortunate feature. The dependence of $A_{1}-A_{-1}$
and $U$ is the source of what the econometrics literature calls ``compliance
heterogeneity.'' For example, the literature on the local ATE (LATE)
and the marginal treatment effect (MTE) essentially builds on this
notion. I believe this issue can be partly mitigated by treating $L$
to be high-dimensional. In fact, both papers' estimation methods allow
for high-dimensional explanatory covariates.\footnote{It is worth noting that even under Assumption 8, $A_{z}$ itself does
not have to be mean independent of $U$. However, $U$ seems hard
to be cancelled out in $E[A_{1}|L,U]-E[A_{-1}|L,U]$ when $A_{z}$
has structure of a binary choice model.}

\citet{cui2019semiparametric} do relax Assumption 8 by introducing
Assumption 7. As noted in their paper, Assumption 8 implies Assumption
7. Also, Condition A5b-1(b) in \citet{qiu2020optimal} implies Assumption
7. Condition A5b-1(b) posits that $E[Y_{1}-Y_{-1}|L,U]=E[Y_{1}-Y_{-1}|L]$,
i.e., the treatment effect $Y_{1}-Y_{-1}$ and the confounder $U$
are conditionally mean independent. This limits the treatment effect
heterogeneity. We definitely want to avoid the additive structure
$Y_{a}=g_{a}(L)+U$, because it significantly restricts \textit{unobserved}
heterogeneity. Also, with binary $Y$ as in the empirical examples
of both papers, it is more natural to assume a nonseparable structure
of binary choice models.

Therefore, I would like to investigate Assumption 7 as it is. Again,
under the causal IV assumption, Assumption 7 can be expressed as
\begin{align*}
Cov\{E[Y_{1}-Y_{-1}|L,U],E[A_{1}-A_{-1}|L,U]|L\} & =0.
\end{align*}
Let us consider a more generic problem after suppressing $L$: for
some generic scalar-valued functions $\delta$ and $\gamma$ of $U$,
consider the assertion that
\begin{align*}
Cov(\delta(U),\gamma(U)) & =0.
\end{align*}
This assertion requires a specific relationship between the distribution
of $U$ and the shape of conditional mean functions $\delta(\cdot)$
and $\gamma(\cdot)$. This is most clearly seen in the following simple
example. Suppose $U\in\{u_{1},u_{2}\}$ is binary (e.g., high and
low unobserved health types) with $\Pr[U=u_{1}]=p_{1}>0$. Then, it
should be that either $\delta$ or $\gamma$ is a constant function.
To see this, let $\delta_{1}\equiv\delta(u_{1})$, $\delta_{2}\equiv\delta(u_{2})$,
$\gamma_{1}\equiv\gamma(u_{1})$, and $\gamma_{2}\equiv\gamma(u_{2})$.
Then,
\begin{align*}
Cov(\delta(U),\gamma(U)) & =p_{1}\delta_{1}\gamma_{1}+(1-p_{1})\delta_{2}\gamma_{2}-\{p_{1}\delta_{1}+(1-p_{1})\delta_{2}\}\{p_{1}\gamma_{2}+(1-p_{1})\gamma_{2}\}\\
 & =p_{1}(1-p_{1})\delta_{1}(\gamma_{1}-\gamma_{2})-p_{1}(1-p_{1})\delta_{2}(\gamma_{1}-\gamma_{2}).
\end{align*}
Therefore, $Cov(\delta(U),\gamma(U))=0$ if and only if $\delta_{1}(\gamma_{1}-\gamma_{2})=\delta_{2}(\gamma_{1}-\gamma_{2})$,
or equivalently, $\delta_{1}=\delta_{2}$ or $\gamma_{1}=\gamma_{2}$,
which proves the claim. This result shows that when $U$ is distributed
as Bernoulli, Assumption 7 is equivalent to Assumption 8 or Condition
A5b-1(b). More generally, $\delta(\cdot)$ and $\gamma(\cdot)$ need
to follow specific shapes given the particular distribution of $U$.
It would be worth investigating how sensitive the estimation and inference
results of \citet{cui2019semiparametric} and \citet{qiu2020optimal}
are when these identifying conditions fail or nearly fail. For example,
one can ask how much suboptimal the estimated regime would become
under the misspecification.

\section{Alternative Identifying Condition for Optimal Regimes}

Inspired by the identification analyses of \citet{cui2019semiparametric}
and \citet{qiu2020optimal}, I would like to propose a simple identifying
condition for optimal dynamic regimes. This condition can be used
as an alternative to Assumption 7 or 8 in \citet{cui2019semiparametric}
and Conditions A5b-1(b) or 2(b) in \citet{qiu2020optimal}. It can
also be an alternative to the identifying conditions in \citet{han2018nonparametric},
which proposes the use of extra exogenous variables besides binary
IVs to identify the dynamic treatment effects and optimal sequential
treatment regimes in longitudinal settings. 

To state the condition, I follow the notation in \citet{cui2019semiparametric}.

\begin{asA}The following two conditions hold for every $L=l$: (a)
either $E[Y_{1}|L,U]\ge E[Y_{-1}|L,U]$ a.s. $[F_{U|L}]$ or $E[Y_{1}|L,U]\le E[Y_{-1}|L,U]$
a.s. $[F_{U|L}]$; (b) either $E[A_{1}|L,U]\ge E[A_{-1}|L,U]$ a.s.
$[F_{U|L}]$ or $E[A_{1}|L,U]\le E[A_{-1}|L,U]$ a.s. $[F_{U|L}]$.

\end{asA}

Assumption A(a) posits that, conditional on $L$, the sign of the
ATE is maintained across individuals defined by unobserved type $U$,
i.e., $\text{sign}\{E[Y_{1}-Y_{-1}|L,U]\}=\text{sign}\{E[Y_{1}-Y_{-1}|L]\}$
a.s. $[F_{U|L}]$. For example, \citet{cui2019semiparametric} mention
in their empirical example (Section 5) that ``one might expect that
having a third child would generally reduce a mother's labor participation
even if the effects are heterogeneous.'' Assumption A(a) can be consistent
with this story. This assumption even allows the following: the effect
of a third child on a mother's labor participation to be positive
for less educated women ($L\le l_{0}$) regardless of their level
of aspiration ($U$), and be negative for highly educated women ($L>l_{0}$)
regardless of their $U$.

A similar interpretation can be made for Assumption A(b) in terms
of the effect of the IV on the treatment decision. This part is weaker
than the LATE monotonicity assumption that $\Pr(A_{1}\ge A_{-1})=1$
(which also appears as Assumption 9 in \citet{cui2019semiparametric}).
This is because when $A_{1}\ge A_{-1}$ a.s. then $E[A_{1}|L,U]\ge E[A_{-1}|L,U]$
a.s. $[F_{U|L}]$, but the converse is not necessarily true.

\begin{theorem}Under Assumptions 2--6 (\citet{cui2019semiparametric})
and Assumption A, $\arg\max_{\mathcal{D}}E[Y_{\mathcal{D}(L)}]$ is
nonparametrically identified by
\begin{align}
\arg\max_{\mathcal{D}}E[Y_{\mathcal{D}(L)}] & =\arg\max_{\mathcal{D}}E\left[\frac{ZI\{A=\mathcal{D}(L)\}YA}{\delta(L)f(Z|L)}\right].\label{eq:ID1}
\end{align}
Furthermore,
\begin{align}
\arg\max_{\mathcal{D}}E[Y_{\mathcal{D}(L)}] & =\arg\max_{\mathcal{D}}E\left[\frac{I\{Z=\mathcal{D}(L)\}Y}{\delta(L)f(Z|L)}\right].\label{eq:ID2}
\end{align}

\end{theorem}

\begin{proof}To prove \eqref{eq:ID1}, by the first and second derivations
in the proof of Theorem 2.1 in \citet{cui2019semiparametric},
\begin{align}
E\left[\frac{ZI\{A=\mathcal{D}(L)\}YA}{\delta(L)f(Z|L)}\right] & =E\left[\tilde{\gamma}(L,U)\tilde{\delta}(L,U)\frac{I\{\mathcal{D}(L)=1\}}{\delta(L)}\right]+E\left[\kappa(L,U)\right],\label{eq:proof_thm1}
\end{align}
where $\kappa(L,U)\equiv\tilde{\delta}(L,U)\frac{E[Y_{-1}|L,U]}{\delta(L)}$,
$\tilde{\delta}(L,U)\equiv\Pr(A=1|Z=1,L,U)-\Pr(A=1|Z=-1,L,U)$, and
$\delta(L)\equiv\Pr(A=1|Z=1,L)-\Pr(A=1|Z=-1,L)$. Recall that, with
$\text{\ensuremath{\Delta}}(L)\equiv E[Y_{1}-Y_{-1}|L]$,
\begin{align}
\arg\max_{\mathcal{D}}E[Y_{\mathcal{D}(L)}] & =\arg\max_{\mathcal{D}}E\left[\text{\ensuremath{\Delta}}(L)I\{\mathcal{D}(L)=1\}\right].\label{eq:proof_thm1_1}
\end{align}
We prove that maximizing the right side of \eqref{eq:proof_thm1}
is equivalent to maximizing $E\left[\text{\ensuremath{\Delta}}(L)I\{\mathcal{D}(L)=1\}\right]$
under Assumption A. Without loss of generality, assume $\delta(L)>0$.
Under Assumption A, for any given $l$,
\begin{align*}
\text{sign}\{\tilde{\gamma}(l,U)\}=\text{sign}\{\text{\ensuremath{\Delta}}(l)\} & \text{ a.s.}
\end{align*}
and
\begin{align*}
\text{sign}\{\tilde{\delta}(l,U)\} & =\text{sign}\{\delta(l)\}=1\text{ a.s.}
\end{align*}
Therefore,
\begin{align}
\text{sign}\{E[\tilde{\gamma}(L,U)\tilde{\delta}(L,U)/\delta(L)|L=l]\} & =\text{sign}\{\text{\ensuremath{\Delta}}(l)\}.\label{eq:proof_thm1_2}
\end{align}
Note that choosing $\mathcal{D}^{*}$ that maximizes $E[\text{\ensuremath{\Delta}}(L)I\{\mathcal{D}(L)=1\}]$
in \eqref{eq:proof_thm1_1} is equivalent to choosing $\mathcal{D}^{*}$
such that $\mathcal{D}^{*}(L)=1$ whenever $\text{\ensuremath{\Delta}}(L)>0$
and $\mathcal{D}^{*}(L)=-1$ otherwise. But by \eqref{eq:proof_thm1_2},
the latter is equivalent to choosing $\mathcal{D}^{*}$ such that
$\mathcal{D}^{*}(L)=1$ whenever $E[\tilde{\gamma}(L,U)\tilde{\delta}(L,U)/\delta(L)|L]>0$
and $\mathcal{D}^{*}(L)=-1$ otherwise. Finally, this is equivalent
to choosing $\mathcal{D}^{*}$ that maximizes
\begin{align*}
E\left[E[\tilde{\gamma}(L,U)\tilde{\delta}(L,U)/\delta(L)|L]I\{\mathcal{D}(L)=1\}\right] & =E\left[\tilde{\gamma}(L,U)\tilde{\delta}(L,U)\frac{I\{\mathcal{D}(L)=1\}}{\delta(L)}\right],
\end{align*}
which completes the proof as $E\left[\kappa(L,U)\right]$ in \eqref{eq:proof_thm1}
does not depend on $\mathcal{D}$.

To prove \eqref{eq:ID2}, by the first derivation in the proof of
Theorem 2.2 in \citet{cui2019semiparametric} (except the last equality),
\begin{align}
E\left[\frac{I\{Z=\mathcal{D}(L)\}Y}{\delta(L)f(Z|L)}\right] & =E\left[\tilde{\gamma}(L,U)\tilde{\delta}(L,U)\frac{I\{\mathcal{D}(L)=1\}}{\delta(L)}\right]+E\left[\kappa(L,U)\right],\label{eq:proof_thm1-1}
\end{align}
where $\kappa(L,U)$ is defined in \citet{cui2019semiparametric}
and does not depend on $\mathcal{D}$. Then, by the same argument
as in the previous case, we have the desired result.\end{proof}

Although Assumption A is helpful to identify the optimal dynamic regime,
it cannot be directly used to identify the value function, whereas
Assumption 8 in \citet{cui2019semiparametric} and Conditions A5b-1(b)
or 2(b) in \citet{qiu2020optimal} are powerful enough to identify
it.\footnote{Assumption A may still have identifying power for the value function
if one is willing to take a partial identification approach.}

\section{Partial Identification Approach}

Although the point identification approach permits powerful inference
on treatment effects and optimal regimes, as is shown in \citet{cui2019semiparametric}
and \citet{qiu2020optimal}, it has to rely on some versions of extrapolative
assumptions. This is inevitable, because when the instrument is binary,
it is known to have no identifying power for general non-compliers,
e.g., always-takers and never-takers, but only identifies the effect
for compliers as the LATE. Therefore, in order to identify the ATE,
which is the effect for the entire population that includes always-takers
and never-takers, and which is a relevant parameter to recover the
optimal regime, it requires means of extrapolation. Broadly speaking,
the assumptions used in \citet{cui2019semiparametric}, \citet{qiu2020optimal},
and \citet{han2018nonparametric} operate in a way that extrapolates
the LATE to the effects for different subpopulations.

Extrapolation inevitably demands prior beliefs from the researcher
or policy maker. Sometimes, one may be interested in knowing how much
she can learn in the absence of extrapolative assumptions and how
sensitive her point identification results are to these assumptions.
This view is shared by \citet{qiu2020optimal} who mention that ``{[}i{]}n
future work, it would be interesting to develop a framework for sensitivity
analyses for Condition A5b.'' In this regard, partial identification
can be a fruitful approach. \citet{cui2019semiparametric} briefly
touch upon this approach. In the Appendix, they derive bounds on the
average potential outcome $E[Y_{\mathcal{D}(L)}]$ for the case of
binary $Y$ using \citet{balke1997bounds}'s bounds and find regimes
that maximize corresponding lower bounds. What could have been done
further is a related sensitivity analysis, e.g., how the estimated
regime under Assumption 7 is consistent with the estimated regime
under their bound analysis.

The partial identification approach to learn optimal dynamic regimes
and welfare is also considered in \citet{han2019optDTR}. In this
paper, I characterize the identified set of the optimal dynamic sequential
regime that maximizes a general form of welfare that nests the average
potential outcome. From data that are generated from multi-period
settings where unmeasured confounders exist, I use an IV method to
establish the sharp partial ordering of welfares (in terms of possible
regimes), the identified set of the optimal dynamic regime, and sharp
bounds on the optimized welfare (i.e., the value function). I also
propose a wide range of identifying assumptions that can tighten the
results and can be easily incorporated within the paper's framework.

Although the partial identification approach may not deliver informative
recommendations for optimal regimes, I believe that it still has value
for policy making for at least three reasons. First, as mentioned
above, one may conduct sensitivity analyses for priors used in point
and partial identification. Second, even with uninformative results,
suboptimal regimes can be easily detected and removed from the policy
menu. Third, the lack of informativeness may guide the researcher
and the policy maker toward better data collection, so that the informativeness
in policy suggestions is not driven from potentially arbitrary assumptions
but from more informative data.

\section{Concluding Remarks}

It is exciting to see works like \citet{cui2019semiparametric} and
\citet{qiu2020optimal} in the literature on individualized treatments
that was pioneered by \citet{murphy2001marginal}, \citet{murphy2003optimal},
and \citet{robins2004optimal} and has grown rapidly since then. The
literature seems to have entered a new territory as it begins to concern
the problem of treatment endogeneity. Many interesting questions can
emerge along the way. I would like to list a few here. First, it would
be interesting to consider the issues of policy invariance in the
context of individualized treatments. For example, it is well known
that different IVs induce different individuals to respond to the
policy. \citet{qiu2020optimal} acknowledge this point when they define
the local average encouragement effect for compliers. Indeed, the
definition of this local parameter raises the question of who the
compliers are, especially because the parameter is defined by contrasting
two different IVs. More generally, it would be important to define
an appropriate---i.e., policy invariant---target population relevant
to the optimized policy, which is inevitably defined by comparing
multiple hypothetical policies. Second, as a way of overcoming the
identification challenge under endogeneity, it would be interesting
to continue developing frameworks that are designed to incorporate
rich data structures, e.g., panels, multiple IVs, and high-dimensional
covariates. Third, related to the second question, extending the online
approach (e.g., reinforcement learning (\citet{shortreed2011informing})
to the context of treatment endogeneity and non-compliance would be
interesting for future work. Finally, the theory of estimation and
inference is worth further exploration. For example, for their asymptotic
theory, \citet{qiu2020optimal} introduce a convergence rate requirement
for the estimated optimal regime that does not seem innocuous. It
would be interesting to investigate whether it is possible to develop
asymptotic theory for the treatment and encouragement effects directly
from the optimization problems (the equations (2) and (3) in their
paper), which may bypass the intermediate estimation of the optimizers
(the optimal regimes).

\bibliographystyle{ecta}
\bibliography{comment_JASA}

\end{document}